\documentclass[10pt]{amsart}
\usepackage{amsmath,amssymb,amsfonts,amsthm}
\usepackage{hyperref}
\usepackage{color}
\usepackage{graphicx}
\usepackage{mathtools}
\usepackage[numbers,sort&compress]{natbib}
\usepackage{lmodern}
\usepackage[T1]{fontenc}

\theoremstyle{plain}
\newtheorem{thm}{Theorem}
\newtheorem{lem}[thm]{Lemma}
\newtheorem{cor}[thm]{Corollary}
\newtheorem{prop}[thm]{Proposition}
\newtheorem{remark}[thm]{Remark}
\newtheorem{defn}[thm]{Definition}
\newtheorem{ex}[thm]{Example}
\numberwithin{thm}{section}
\numberwithin{equation}{section}

\newcommand{\defas}{\mathrel{\mathop:}=} % :=
\newcommand{\set}[1]{\{#1\}}
\newcommand{\norm}[1]{\|#1\|}
\newcommand{\abs}[1]{|#1|}
\newcommand{\Ext}{\mathrm{ext}}
\newcommand{\Int}{\mathrm{int}}
\newcommand{\zz}{{\overline{z}}}
\newcommand{\Li}{\mathrm{Li}}
\DeclareMathOperator{\RE}{Re}
\DeclareMathOperator{\IM}{Im}
\newcommand{\dd}{\mathrm{d}}
\newcommand{\ldeg}[1]{\underline{\operatorname{deg}}_{#1}}

\newcommand{\sE}{{\mathcal E}}
\newcommand{\sF}{{\mathcal F}}
\newcommand{\sI}{{\mathcal I}}
\newcommand{\sV}{{\mathcal V}}
\newcommand{\CC}{{\mathbb C}}
\newcommand{\NN}{{\mathbb N}}
\newcommand{\PP}{{\mathbb P}}
\newcommand{\QQ}{{\mathbb Q}}
\newcommand{\RR}{{\mathbb R}}

\title{Graphical functions in parametric space}
\author{Marcel Golz, Erik Panzer and Oliver Schnetz}
\begin{document}
\begin{abstract}
Graphical functions are positive functions on the punctured complex plane $\CC\setminus\set{0,1}$ which arise in quantum field theory. We generalize a parametric integral representation for graphical functions due to Lam, Lebrun and Nakanishi, which implies the real analyticity of graphical functions.
Moreover we prove a formula that relates graphical functions of planar dual graphs.
\end{abstract}
\maketitle

\section{Introduction}
One main problem in perturbative quantum field theory is the calculation of Feynman integrals (see e.g.\ \cite{ItzyksonZuber}). As a new tool for this task, graphical functions were introduced by the third author in \cite{GraphicalFunctions}. Basically, these are special classes of massless Feynman integrals (3-point functions) that can be understood as single-valued functions on the punctured complex plane $\CC\setminus\set{0,1}$.
They are powerful tools in multi-loop calculations, see e.g.\ \cite{ZigZag,Coaction}.

A traditional method to study Feynman integrals is to represent them in a parametric version, where one integrates over variables associated to the edges of a Feynman graph \cite{ItzyksonZuber}. In many cases of interest, these integrals can be computed in terms of multiple polylogarithms, using a method developed by F.~Brown \cite{Brown:SomeFI,Brown:TwoPoint} and the second author \cite{Panzer:HyperInt,Panzer:PhD}. The combination of graphical functions and this parametric integration (using the formulas derived in this article) has recently provided a breakthrough in the calculation of primitive log-divergent amplitudes of graphs with up to eleven independent cycles (`loops') \cite{Coaction}.

In a complete quantum field theoretical calculation one encounters naive singularities which are most frequently treated by the `dimensional regularization scheme' which demands the generalization to arbitrary space-time dimensions (away from the classical four dimensions). The parametric representation is the cleanest way to define Feynman integrals in non-integer `dimensions'.
In this article, we derive fundamental formulas and results for graphical functions in parametric representations for arbitrary dimensions.

Apart from \cite{Coaction}, first applications of the results of this article include the calculation of the beta function and field anomalous dimension of minimally subtracted $(4-\varepsilon)$-dimensional $\phi^4$ theory to six and seven loops by the third author \cite{Schnetz6loops}.

\subsection{Feynman integrals in position space}
A Feynman graph is a graph $G$ with a distinguished subset $\sV_G^\Ext \subseteq \sV_G$ of \emph{external} vertices (the remaining vertices $\sV_G^\Int = \sV_G\setminus \sV_G^\Ext$ are called
\emph{internal}). We often suppress the subscript $G$ and we use roman capital letters for cardinalities, so e.g.\ $\sV^\Ext=\sV_G^\Ext$ and $V^\Ext=\abs{\sV^\Ext}$.
We fix the dimension\footnote{%
In two dimensions ($\lambda=0$), non-trivial graphical functions \eqref{eq:gf-xspace} always diverge. However one can redefine $\lambda\nu_e=:\nu_e'\in\RR$ as the edge weights and all of the following results extend to this case.}
\begin{equation*}
	d = 2\lambda+2 > 2
\end{equation*}
and associate to every vertex $v$ of $G$ a $d$-dimensional vector $x_v\in\RR^d$.
An edge $e$ between vertices $u$ and $v$ corresponds to the quadratic form $Q_e$ which is the square of the Euclidean distance between $x_u$ and $x_v$,
\begin{equation}\label{eq:euclidean-norm}
	Q_e
	=\norm{x_u-x_v}^2
	=\sum_{i=1}^d (x_u^i - x_v^i)^2.
\end{equation}
Moreover, every edge $e$ has an edge weight $\nu_e\in\RR$.
Then the Feynman integral associated to $G$ in position space is defined as
\begin{equation}\label{eq:gf-xspace}
	f_G^{(\lambda)}(x)
	=\left(\prod_{v\in\sV^{\Int}}\int_{\RR^d}\frac{\dd^dx_v}{\pi^{d/2}}\right)
	\frac{1}{\prod_e Q_e^{\lambda\nu_e}},
\end{equation}
where the first product is over all internal vertices of $G$ and the second product is over all edges of $G$. Note that $f_G^{(\lambda)}(x)$ is a function of the external vectors $x=(x_v)_{v\in\sV^\Ext}$ which we always assume to be pairwise distinct ($x_v\neq x_w$ for $v\neq w$).

The convergence of \eqref{eq:gf-xspace} is equivalent to two conditions named `infrared' and `ultraviolet' (this weighted analog of \cite[Lemma~3.4]{GraphicalFunctions} rests on \emph{power counting} \cite{LowensteinZimmermann}):
\begin{itemize}
	\item The graph $G$ is called \emph{ultraviolet convergent} if
\begin{equation}\label{ultraviolet}
	\lambda \nu_g < \tfrac{d}{2} (V_g-1)
\end{equation}
holds for all induced%
\footnote{A subgraph $g$ is induced when every edge of $G$ which has both endpoints in $\sV_g$ belongs to $g$.}
subgraphs $g$ with $\abs{\sV_g \cap \sV^\Ext} \leq 1$.
Here we write
\begin{equation*}
	\nu_g
	=\sum_{e\in\sE_g} \nu_e
\end{equation*}
and denote the sets of vertices and edges of $g$ with $\sV_g$ and $\sE_g$.

	\item
		A vertex $v\in\sV_g$ of a subgraph $g$ of $G$ is called $g$-internal if it is internal ($v \in \sV^\Int$) and all edges of $G$ which are incident to $v$ also belong to $g$. We write $V_g^\Int$ for the number of such vertices.
		The graph $G$ is called \emph{infrared convergent} if
\begin{equation}\label{infrared}
	\lambda\nu_g > \tfrac{d}{2} V_g^\Int
\end{equation}
		holds for all subgraphs $g$ of $G$ which satisfy $V_g^\Int>0$ and contain only edges which are incident to at least one $g$-internal vertex.
\end{itemize}
\begin{ex}\label{ex:g4-convergence}
	In case of the graph $G_4$ from figure~\ref{fig:g4g7}, there are three ultraviolet conditions of the form $\lambda \nu_e< \frac{d}{2}$ (one for each edge $e$) and one infrared condition $\lambda \nu_{G_4} > \frac{d}{2}$ (from the full subgraph $g=G_4$).
\end{ex}

\subsection{Graphical functions}
In the special case of three external vertices, we label them with $0$, $1$ and $z$. 
Note that $f_G^{(\lambda)}$ is invariant under the Euclidean group, so we may translate $x_0$ to the origin and rotate $x_1$ and $x_z$ into the plane $\RR^2\times\set{0}^{d-2}$ which we identify with the complex numbers $\CC$. The \emph{graphical function}
\begin{equation*}
	f_G^{(\lambda)}(z)\colon
	\CC\setminus\set{0,1} \longrightarrow \RR_+
\end{equation*}
is a parametrization of $f_G^{(\lambda)}(x)$ defined in terms of a complex variable $z\neq 0,1$ via
\begin{equation}\label{externalvectors}
	x_0 = (0,\ldots,0)^t,
	\quad
	x_1 = (1,0,\ldots,0)^t
	\quad\text{and}\quad
	x_z = (\RE z,\IM z,0,\ldots,0)^t.
\end{equation}
Graphical functions were introduced in \cite{GraphicalFunctions} basically as a tool for calculating Feynman periods in $\phi^4$ quantum field theory (see also \cite{Coaction,Drummond:Ladders,Todorov}).
However, they can also appear in amplitudes and correlation functions, see for example \cite{OffShellConformal}.

In \cite{GraphicalFunctions} `completions' of graphical functions were defined. In this article, however, we use uncompleted graphs.
\begin{figure}%
	\includegraphics{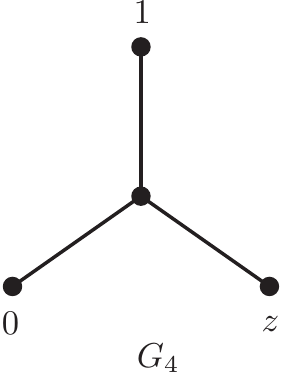} \hspace{3cm}
	\includegraphics{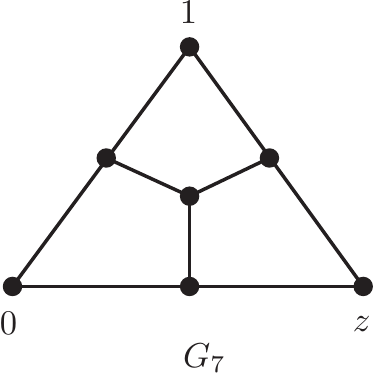}
	\caption{Examples of connected graphs with four and seven vertices in total and three external vertices labeled $0$, $1$ and $z$.}%
	\label{fig:g4g7}%
\end{figure}
\begin{ex}\label{ex:g4}
	In $d=4$ dimensions and with edge weights $\nu_e=1$, the graph $G_4$ of Figure~\ref{fig:g4g7} has a convergent graphical function (see example~\ref{ex:g4-convergence}). It is (see \cite{GraphicalFunctions,Todorov})
\begin{equation*}
	f_{G_4}^{(1)}(z)
	= \int_{\RR^4} \frac{\dd^4 x}{\pi^2} \frac{1}{\norm{x}^2 \norm{x-x_1}^2 \norm{x-x_z}^2}
	= \frac{4{\mathrm i}D(z)}{z-\zz}
\end{equation*}
in terms of the Bloch-Wigner dilogarithm 
$ D(z)=\IM (\Li_2(z)+\log(1-z)\log\abs{z}) $.
\end{ex}
The Bloch-Wigner dilogarithm $D(z)$ is a single-valued version of the dilogarithm $\Li_2(z)=\sum_{k=1}^\infty z^k/k^2$. It is real analytic on $\CC\setminus\set{0,1}$ and antisymmetric under complex conjugation $D(z)=-D(\zz)$. These properties of the Bloch-Wigner dilogarithm lift to general properties of graphical functions:
\begin{thm}\label{generalthm}
 Let $G$ be a graph which fulfills the ultraviolet and infrared conditions \eqref{ultraviolet} and \eqref{infrared}.
 Then the graphical function
$f_G^{(\lambda)}\colon \CC\setminus\set{0,1} \longrightarrow \RR_+$
has the following general properties:
\begin{enumerate}
	\item[(G1)] $f_G^{(\lambda)}(z)=f_G^{(\lambda)}(\zz)$,
	\item[(G2)] $f_G^{(\lambda)}$ is single-valued and
	\item[(G3)] $f_G^{(\lambda)}$ is real analytic on $\CC\setminus\{0,1\}$.
\end{enumerate}
\end{thm}
It was not possible to prove real analyticity (G3) in full generality with the methods in \cite{GraphicalFunctions}.
In this article we obtain (G3) as a consequence of an alternative integral representation of graphical functions. In this representation, the integration variables $\alpha_e$ (known as \emph{Schwinger} or \emph{Feynman} parameters) are associated to edges of the graph \cite{ItzyksonZuber,BEK}.

Although we are mainly interested in the case of three external vertices 0, 1, $z$, our results effortlessly generalize to an arbitrary number $V^\Ext$ of external vertices.

\subsection{Graph polynomials}
We will use certain polynomials in the edge variables $\alpha_e$ that were defined and studied by F.~Brown and K.~Yeats \cite{BrownYeats:WD}.
\begin{defn}\label{spanningforest}
Let $p=\set{p_1,\ldots ,p_n}$ denote a partition of a subset of the vertices of a graph $G$ (so $p_i\subseteq \sV$ and $p_i\cap p_j = \emptyset$ when $i\neq j$). We write $\sF_G^p$ for the set of all spanning forests $T_1\cup\ldots\cup T_n$ consisting of exactly $n$ (pairwise disjoint) trees $T_i$ such that $p_i \subseteq T_i$.
The \emph{dual spanning forest polynomial} associated to $p$ is
\begin{equation}\label{eq:dual-forpol}
	\tilde{\Psi}_G^p(\alpha)
	\defas\sum_{F\in\sF_G^p}\;\prod_{e\in F}\alpha_e.
\end{equation}
We suppress curly brackets in the notation, so for example $\tilde{\Psi}_G^{01z} = \tilde{\Psi}_G^{\set{\set{0,1,z}}}$ denotes the sum of spanning forests ($n=1$), while the partition in $\tilde{\Psi}_G^{01,z}$ is $\set{\set{0,1},\set{z}}$ ($n=2$).
Say we call the external vertices $1,\ldots,V^\Ext$, then we write
$\tilde{\Psi} \defas \tilde{\Psi}^{1,\ldots,V^\Ext}$
for the partition into singletons ($n = V^\Ext$). The partitions with $n=V^\Ext-1$ have exactly one part containing two external vertices. We collect them in the polynomial
\begin{equation}\label{eq:dual-phi}
	\tilde{\Phi}_G(\alpha,x)
	\defas\sum_{1\leq i<j\leq V^{\Ext}} 
	\norm{x_i-x_j}^2
	\tilde{\Psi}_G^{ij,(k)_{k\neq i,j}}(\alpha) 
	.
\end{equation}
\end{defn}
%The spanning forest polynomial $\Psi_G^{01z}$ is the graph polynomial $\Psi_G$ while the spanning forest polynomial
%$\Psi_G^{0,1,z}$ equals the graph polynomial $\Psi_{G/\Ext}$ of the graph $G/\Ext$ that one obtains from $G$ by identifying the three external vertices without changing the edge labels.
\begin{ex}
	If we label the three edges adjacent to $0$, $1$ and $z$ in $G_4$ (see Figure~\ref{fig:g4g7}) by $1$, $2$ and $3$, then we find the polynomials
\begin{align*}
%	\Psi_{G_4}^{1z,0} &= \alpha_1,
%	&
%	\Psi_{G_4}^{01z} &= 1,%\Psi_{G_4} =1,
%	\\
%	\Psi_{G_4}^{0z,1} &= \alpha_2,
%	&
%	\Psi_{G_4}^{0,1,z}
%	&= \alpha_2\alpha_3+\alpha_1\alpha_3+\alpha_1\alpha_2,
%	\\
%	\Psi_{G_4}^{01,z} &= \alpha_3,
%	&
%	\Phi_{G_4}
%	&= \alpha_1(z-1)(\zz-1)+\alpha_2z\zz+\alpha_3.
	\tilde{\Psi}_{G_4}^{1z,0} &= \alpha_2 \alpha_3,
	&
	\tilde{\Psi}_{G_4}^{01z} &= \alpha_1 \alpha_2 \alpha_3,%\Psi_{G_4} =1,
	\\
	\tilde{\Psi}_{G_4}^{0z,1} &= \alpha_1 \alpha_3,
	&
	\tilde{\Psi}_{G_4}^{0,1,z}
	&= \alpha_1+\alpha_2+\alpha_3,
	\\
	\tilde{\Psi}_{G_4}^{01,z} &= \alpha_1 \alpha_2,
	&
	\tilde{\Phi}_{G_4}
	&= (z-1)(\zz-1)\alpha_2 \alpha_3+z\zz \alpha_1 \alpha_3 + \alpha_1 \alpha_2.
\end{align*}
Here $\zz$ denotes the complex conjugate of $z\in\CC\setminus\set{0,1}$ and we used \eqref{externalvectors}.
\end{ex}
A parametric (i.e.\ depending on the edge parameters $\alpha_e$) formula for (massive) position space Feynman integrals in four-dimensional Minkowski space was discovered long ago \cite{Nakanishi:Hankel,LamLebrun} and is also discussed in the book \cite[Equation~(8-33)]{Nakanishi:Book}. In the massless, Euclidean case it becomes a parametric formula for graphical functions. We give an extension to arbitrary dimensions which also allows for negative edge weights.\footnote{The validity for arbitrary dimensions is straightforward and was noticed already in \cite[Remark~7-10]{Nakanishi:Book}.}
\begin{thm}\label{dualparam}
Let $G$ be a non-empty graph with $V^{\Int}_G$ internal vertices and edges labeled $1,2,\ldots,E_G$.
We assume that its graphical function \eqref{eq:gf-xspace} converges, meaning that $G$ is subject to \eqref{ultraviolet} and \eqref{infrared}, and define the \emph{superficial degree of divergence}
\begin{equation}\label{defM}
	M_G \defas \lambda \nu_G-\tfrac{d}{2} V^\Int_G.
\end{equation}
Then for any set of non-negative integers $n_e$ such that $n_e+\lambda\nu_e>0$, we have the following \emph{dual parametric representation} of $f_G^{(\lambda)}$ as a convergent projective integral:
\begin{equation}\label{fdualparam}
	f_G^{(\lambda)}(x)
	=
	\frac{(-1)^{\sum_e\!n_e}\,\Gamma(M_G)}{\prod_e\Gamma(n_e+\lambda\nu_e)}
	\int_\Delta \Omega 
	\Big[\prod_e\alpha_e^{n_e+\lambda\nu_e-1}\partial_{\alpha_e}^{n_e}\Big]
	\frac{1}{\tilde{\Phi}_G^{M_G}\tilde{\Psi}_G^{d/2-M_G}},
\end{equation}
where the integration domain is given by the positive coordinate simplex
\begin{equation*}
	\Delta
	=\{(\alpha_1:\alpha_2:\ldots:\alpha_{E_G})\colon \alpha_e>0\text{ for all }e\in\{1,2,\ldots,E_G\}\}\subset\PP^{E_G-1}\RR
\end{equation*}
and we set
\begin{equation*}
	\Omega
	=\sum_{e=1}^{E_G}(-1)^{e-1}\alpha_e \dd\alpha_1\wedge\ldots\wedge\widehat{\dd\alpha_e}\wedge\ldots\wedge \dd\alpha_{E_G}
	.
\end{equation*}
\end{thm}
\begin{remark}
For integer $\lambda\nu_e \leq 0$ one may set $n_e=1-\lambda\nu_e$ such that the integration over $\alpha_e$ trivializes to the evaluation at $\alpha_e = 0$ of a $(-\lambda\nu_e)$'s derivative.
\end{remark}
Readers who are not familiar with projective integrals can specialize to an affine integral by setting $\alpha_1=1$ and integrating the remaining $\alpha_e$ ($e>1$) from $0$ to $\infty$.

Note that $M_G$ is restricted by convergence: From \eqref{infrared} with $g=G$ and from \eqref{ultraviolet} with $g=G\setminus(\sV^\Ext\setminus\set{v})$ (for some $v\in\sV^\Ext$) we obtain for a graph $G$ with no edges between external vertices that
\begin{equation*}
	0<M_G<\lambda \min_{v\in V^\Ext} \sum_{w\in\sV^\Ext\setminus\set{v}} \nu_w,
\end{equation*}
where $\nu_w$ is the sum of weights $\nu_e$ of all edges $e$ adjacent to the external vertex $w$.

One immediate advantage of the parametric representation is that for many graphs with not more than nine vertices, the integral \eqref{fparam} can be calculated (in terms of polylogarithms) with methods developed by F.~Brown \cite{Brown:SomeFI} and the second author \cite{Panzer:PhD,Panzer:HyperInt}.

Note that we obtain another integral representation via the Cremona transformation $\alpha_e \rightarrow 1/\alpha_e$:
\begin{cor}\label{paramthm}
Let $G$ be a non-empty graph with $E_G$ edges.
We assume the convergence of $f_G^{(\lambda)}$ and also that every edge $e$ has a positive weight $\nu_e>0$. Then
\begin{equation}\label{fparam}
	f_G^{(\lambda)}(x)
	=
	\frac{\Gamma(M_G)}{\prod_{e} \Gamma(\lambda\nu_e)}
	\int_\Delta \frac{
		\prod_{e} \alpha_e^{d/2-\lambda \nu_e  - 1}
	}{
		\Phi_G^{M_G}\Psi_G^{d/2-M_G}
	}\Omega,
\end{equation}
where $\Psi_G = \Psi_G^{1,\ldots,V^\Ext}$ and
$ \Phi_G(\alpha,x) = \sum_{i<j} \norm{x_{i} - x_{j}}^2 \Psi_G^{i j, (k)_{k\neq i,j}}(\alpha)$ are defined in terms of the spanning forest polynomials, which are dual to \eqref{eq:dual-forpol}:
\begin{equation}\label{eq:forpol}
	\Psi_G^p(\alpha)
	=\sum_{F\in\sF_G^p} \prod_{e\notin F}\alpha_e
	=\Big(\prod_e\alpha_e\Big)\tilde{\Psi}_G^p(\alpha^{-1}).
\end{equation}
\end{cor}
\begin{proof} We set $n_e=0$ in \eqref{fdualparam} for all edges $e$ of $G$. We use the affine chart $\alpha_1=1$ in \eqref{fdualparam} and invert all $\alpha_e$, $e>1$. By \eqref{eq:forpol} this gives the integrand in \eqref{fparam} for $\alpha_1=1$. The projective version of this integral is \eqref{fparam}.
\end{proof}

\subsection{Planar duals}
A planar dual $G^\star$ of a Feynman graph $G$ with external vertices $0,1,z$ is a usual planar dual graph to which we add external vertices at `opposite' sides, see Figure~\ref{fig:duals} (a precise description will be given in Definition~\ref{dualdef}).
\begin{figure}
	\includegraphics{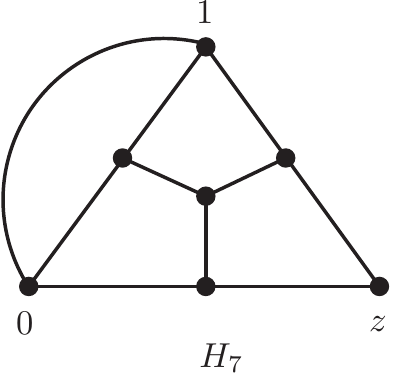} \hspace{2cm}
	\includegraphics{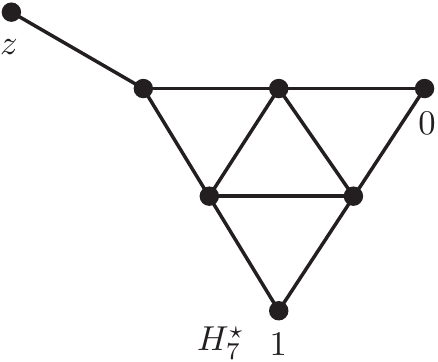}
	\caption{The graphs $H_7^{}$ and $H_7^\star$ are planar duals.}%
	\label{fig:duals}%
\end{figure}
In the case when $M_G=d/2$, graphical functions of dual graphs are related:
\begin{thm}\label{planarthm}
Let $G$ be a connected graph with external vertices $0,1,z$ and edge weights $\nu_e>0$ such that the graphical function $f_G^{(\lambda)}$ converges and
$M_G=d/2$.
Let $G^\star$ be a dual of $G$ and denote by $e^\star$ the edge of $G^\star$ which corresponds to the edge $e$ of $G$.
Let the edge weights $\nu_{e^\star}$ of $G^\star$ be related to the edge weights $\nu_e$ of $G$ through
\begin{equation}\label{weightsstar}
	\lambda\nu_{e^\star}
	=d/2-\lambda\nu_e.
\end{equation}
Then the graphical functions associated to $G$ and $G^\star$ are multiples of each other:
\begin{equation}\label{eq:planar-dual}
	f_{G^\star}^{(\lambda)}(z)
	=
	f_G^{(\lambda)}(z)
	\prod_e \frac{\Gamma(\lambda\nu_e)}{\Gamma(\lambda\nu_{e^\star})}
	.
\end{equation}
\end{thm}
Note that ultraviolet convergence (\ref{ultraviolet}) for a single edge $e$ implies $\lambda\nu_e<d/2$, thus $\nu_e^\star>0$. Similarly, positive edge weights in $G$ ensure that the dual graphical function $f_{G^\star}^{(\lambda)}$ is ultraviolet convergent for each single edge $e^\star$ of $G^\star$.
The convergence of $f_{G^\star}^{(\lambda)}$ is ensured by the proof of Theorem~\ref{planarthm}.

If in four dimensions a graph $G$ has edge weights 1 then a dual graph $G^\star$ has also edge weights 1 and the graphical functions are equal if $M_G=2$.

One can also use duality for a planar graph $G$ with $M_G\neq d/2$ if one adds an edge from 0 to 1 of weight $(d/2-M_G)/\lambda$, see the subsequent example~\ref{ex:dual}.
\begin{remark}
	It is well known (see \cite{MinamiNakanishi} for example) that the graphical function of every planar graph $G$ (without restrictions on $\nu_e$ and $d$) is related (by a constant factor) to the \emph{momentum space} Feynman integral associated to $G^{\star}$. What makes Theorem~\ref{planarthm} interesting is that in the particular case when $V^\Ext=3$ and $M_G=d/2$, the momentum- and position space Feynman integrals coincide.
\end{remark}
\begin{ex}\label{ex:dual}
	We want to calculate the $4$-dimensional graphical function of the graph $G_7$ in Figure~\ref{fig:g4g7} with unit edge weights, so $M_{G_7}=1$. To apply Theorem~\ref{planarthm} we add an edge between 0 and 1 (see Figure~\ref{fig:duals}). This does not change the graphical function $f_{G_7}^{(1)}=f_{H_7}^{(1)}$, which is clear from \eqref{eq:gf-xspace}. Theorem~\ref{planarthm} gives $f_{H_7}^{(1)}=f_{H^\star_7}^{(1)}$.
The graphical function of $H^\star_7$ can be calculated by the techniques completion and appending of an edge \cite[Sections 3.4 and 3.5]{GraphicalFunctions}. We obtain
\begin{equation*}
	f_{G_7}^{(1)} (z)
	= 20\zeta(5) \frac{4{\mathrm i}D(z)}{z-\zz},
\end{equation*}
where $\zeta(s)=\sum_{k=1}^\infty k^{-s}$ is the Riemann zeta function.
\end{ex}

\begin{ex}
One obtains a self dual graph $H_4=H_4^{\star}$ with $M_{H_4}=2$ if one adds an edge from 0 to 1 to $G_4$. In this case planar duality leads to a trivial statement.
\end{ex}

\subsection*{Acknowledgements.}
Parts of this article were written while Erik Panzer and Oliver Schnetz were visiting scientists at Humboldt University, Berlin.

\section{proof of Theorem~\ref{dualparam}}
Our proof follows the Schwinger trick (see e.g.\ \cite{ItzyksonZuber}). 
From the definition of the gamma function we obtain for $n+\lambda\nu>0$ the convergent integral (note $Q_e>0$)
\begin{equation}\label{trick}
	\frac{1}{Q_e^{\lambda\nu_e}}
	=\frac{1}{\Gamma(n_e+\lambda\nu_e)}\int_0^\infty\alpha_e^{n_e+\lambda\nu_e-1}(-\partial_{\alpha_e})^n\exp(-\alpha_e Q_e) \ \dd\alpha_e.
\end{equation}
We use this formula to replace the product of propagators in \eqref{eq:gf-xspace} by an integral over the edge parameters $\alpha_e$. Since the integrand $\prod_e \big[ \alpha_e^{n_e + \lambda \nu_e - 1} Q_e^{n_e} \exp(-\alpha_e Q_e) \big]$ is positive, the integral is absolutely convergent and we may interchange the order of integration by Fubini's theorem.
In fact, we can also interchange{\footnotemark} the integration over the vertex variables with the partial derivatives $\partial_{\alpha_e}$ to obtain
\begin{equation}\label{f1}
	f_G^{(\lambda)}(x)
	=\frac{1}{\prod_e\Gamma(n_e+\lambda\nu_e)}
	\int_0^\infty\!\!\!\cdots\int_0^\infty
	\Big[\prod_e\alpha^{n_e+\lambda\nu_e-1}(-\partial_{\alpha_e})^{n_e}\Big]
	\sI(\alpha)\prod_e\dd\alpha_e,
\end{equation}
\footnotetext{%
	We can invoke Theorem~\ref{holomorphic} because $\sI(\alpha')$ is finite for $\alpha_e'>0$ (as we will show) and majorizes $\sI(\alpha)$ (on the integrand level) for all $\alpha$ such that $\alpha_e>\alpha_e'$ for all edges $e$. Note that under the assumptions of Theorem~\ref{holomorphic}, $\partial_{\alpha_e} \sI(\alpha)$ coincides with differentiation under the integral sign (see \cite{Mattner} and \cite[Satz~5.8]{Elstrodt}.}%
where $\sI(\alpha)$ is the Gau{\ss}ian integral
\begin{equation*}
	\sI(\alpha)
	=\left(\prod_{v\,{\rm internal}}\int_{\RR^d}\frac{\dd^dx_v}{\pi^{d/2}}\right)\exp\left(-\sum_e \alpha_eQ_e\right).
\end{equation*}
It factorizes into $d$ parts $\sI_k$, one for each coordinate $k$, since the quadratic form \eqref{eq:euclidean-norm} is diagonal.
We arrange the $i$th coordinates of the $V_G$ vertex variables to the vector $(x_{\Int},x_{\Ext})^t$ where $x_{\Int}=(x_v^k)_{v\in\sV^\Int}$ and $x_{\Ext}=(x_v^k)_{v\in\sV^{\Ext}}$.
Then, the quadratic form in the exponential of $\sI_k$ takes the form
\begin{equation*}
	\sum_e\alpha_eQ_e^k
	=x^t_{\Int} L^{\mathrm{ii}}(\alpha) x_{\Int}
	+x^t_{\Int} L^{\mathrm{ie}}(\alpha) x_{\Ext}
	+x_{\Ext}^t L^{\mathrm{ei}}(\alpha) x_{\Int}
	+x_{\Ext}^t L^{\mathrm{ee}}(\alpha) x_{\Ext}
\end{equation*}
in terms of the (symmetric) Laplace matrix \cite{BognerWeinzierl}
\begin{equation}\label{Ldef}
	L
	= \begin{pmatrix}
		L^{\mathrm{ii}} & L^{\mathrm{ie}} \\
		L^{\mathrm{ei}} & L^{\mathrm{ee}} \\
	\end{pmatrix}
	\quad\text{with entries}\quad
	L(\alpha)_{uv}
	=\begin{cases}
		\sum\limits_{e{\rm\,incident\,to\,}v}\alpha_e & \text{if $u=v$ and}\\
		-\sum\limits_{e=\{u,v\}}\alpha_e & \text{otherwise.}
	\end{cases}
\end{equation}
By convergence, $L^{\mathrm{ii}}$ is positive definite.
We complete the quadratic form to a perfect square, shift the integration variable to $x_{\Int}+L^{\mathrm{ii}-1}L^{\mathrm{ie}}x_{\Ext}$ and obtain by a standard calculation
\begin{equation*}
	\sI_k
	= 
	\det(L^{\mathrm{ii}})^{-1/2}
	\exp\Big(
		x_{\Ext}^t
		[L^{\mathrm{ei}}L^{\mathrm{ii}-1}L^{\mathrm{ie}}-L^{\mathrm{ee}}]
		x_{\Ext}
	\Big).
\end{equation*}
The summation over $k$ in the exponent therefore leads us to
\begin{equation}\label{temp}
	\sI(\alpha)
	= \prod_{k=1}^d \sI_k(\alpha)
	= \det(L^{\mathrm{ii}})^{-d/2}
	\exp\Big(
		\sum_{k,\ell=1}^{V^\Ext}
		(x_k^t x_\ell^{})
		[L^{\mathrm{ei}}L^{\mathrm{ii}-1}L^{\mathrm{ie}}-L^{\mathrm{ee}}]_{k,\ell}
	\Big).
\end{equation}
An application of the matrix tree theorems \cite{Chaiken,BognerWeinzierl} shows that\footnote{%
	In the notation of the (All minors) matrix tree theorem \cite[equation~(2)]{Chaiken}, the first equality is precisely the case $W=U=V^\Ext$, $S=V$. The second identity follows from Cramer's rule by setting $W=V^\Ext \cup \set{v}$ and $U=V^\Ext \cup \set{w}$ and noting that $\varepsilon(W,S)\varepsilon(U,S) = (-1)^{v+w}$ by the remarks after \cite[equation~(3)]{Chaiken}.
}
\begin{equation*}
	\det(L^{\mathrm{ii}})=\tilde{\Psi}
	\quad\text{and}\quad
	\left(L^{\mathrm{ii}-1}\right)_{v,w}
	=\frac{1}{\tilde{\Psi}}\tilde{\Psi}_G^{vw,1,\ldots,V^{\Ext}}
\end{equation*}
for all internal $v$ and $w$.
We can therefore interpret the matrix elements
\begin{equation}\label{term1}
	\tilde{\Psi}(L^{\mathrm{ei}}L^{\mathrm{ii}-1}L^{\mathrm{ie}})_{k,\ell}
	=
	\sum_{\substack{e=\set{k,v}\\f=\set{\ell,w}}}
		\alpha_e\alpha_f
		\tilde{\Psi}_G^{vw,1,\ldots,V^{\Ext}}
	\quad\text{($v, w$ internal)}
\end{equation}
in the exponential of \eqref{temp} in terms of subgraphs of $G$. We distinguish two cases:
\begin{figure}%
	$k\neq \ell$:\ \ \raisebox{-0.5\height}{\includegraphics[height=3.5cm]{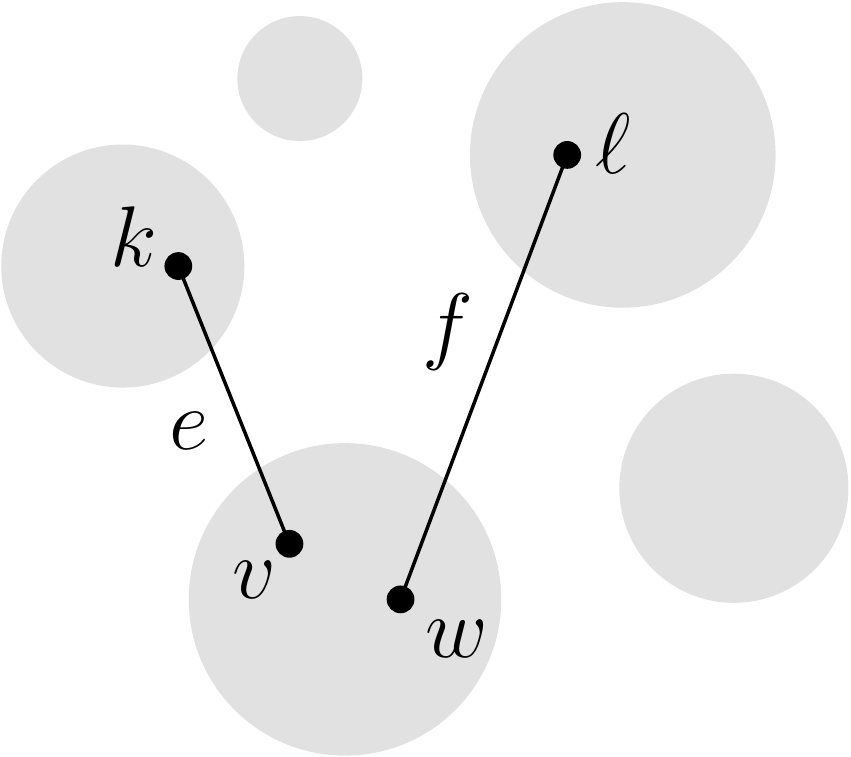}}
	\hspace{1cm}
	$k=\ell$:\ \ \raisebox{-0.5\height}{\includegraphics[height=3.5cm]{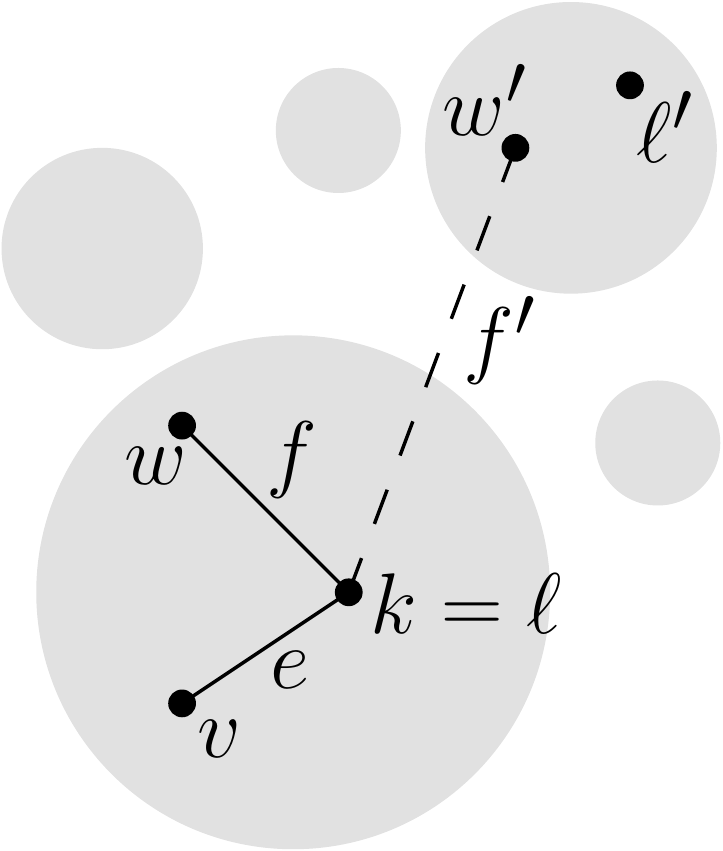}}%
	\caption{%
		For $k\neq \ell$, the grey areas indicate the connected components of $F$. Adding $e$ and $f$ connects $k$ with $\ell$.
		In the case $k=\ell$ we depict the connected components of $F'=F\cup\set{e}$; note that $w$ lies in the same component as $k$.
		When we extend the sum to all edges $f$ incident to $k$, additional contributions arise when $w$ lies in a different
		component of $F'$ and thus connects $k$ to another external vertex $\ell'$ (indicated by the dashed edge $f'$).%
}%
	\label{fig:forests} %
\end{figure}
\begin{itemize}
\item[$k\neq \ell$:] 
	Adding the two edges $e$, $f$ to a spanning forest $F \in \sF_G^{vw,1,\ldots,V^\Ext}$ yields a forest $F' = F \cup \set{e,f} \in \sF_G^{k\ell,(m)_{m\neq k,\ell}}$
	(see Figure~\ref{fig:forests}). Conversely, each $F'$ arises exactly once this way, because it determines $e$ and $f$ as the initial and final edges on the unique path in $F'$ connecting $k$ and $\ell$.
	The only exception are forests $F'$ where this path is just a single edge $e=\set{k,\ell}$ connecting them directly. But in this case $F'\setminus e \in \sF_G^{1,\ldots,V^{\Ext}}$, so we conclude
\begin{equation*}
	\sum_{\mathclap{\substack{e=\set{k,v}\\f=\set{\ell,w}}}}
	\alpha_e\alpha_f\tilde{\Psi}_G^{vw,1,\ldots,V^{\Ext}}(\alpha)
	=\tilde{\Psi}_G^{k\ell,(m)_{m\neq k,\ell}}(\alpha)
	-\tilde{\Psi} \sum_{\mathclap{e=\set{k,\ell}}} \alpha_e
	.
\end{equation*}

\item[$k=\ell$:] Adding $e$ to $F \in \sF_G^{vw,1,\ldots,V^\Ext}$ gives a forest $F'=F \cup \set{e} \in \sF_G^{1,\ldots,kw,\ldots,V^{\Ext}}$. Each such $F'$ occurs exactly once, because $e$ is necessarily the (unique) first edge on the path in $F'$ connecting $k$ with $w$, hence
	\begin{equation*}
		(\tilde{\Psi}L^{\mathrm{ei}}L^{\mathrm{ii}-1}L^{\mathrm{ie}})_{k,k}
		= \sum_{f=\set{k,w}} \alpha_f \tilde{\Psi}_G^{1,\ldots,kw,\ldots,V^\Ext}
		.
	\end{equation*}
	For a fixed $F'\in\sF_G^{1,\ldots,kw,\ldots,V^\Ext}$, $f$ runs over all edges that connect $k$ to a vertex $w$ that lies in the same connected component of $F'$. If we sum instead over all edges incident to $k$, we get additional contributions when $w$ lies in another component, say the one containing $\ell'$ (see Figure~\ref{fig:forests}). Therefore,
	\begin{equation*}
		(\tilde{\Psi}L^{\mathrm{ei}}L^{\mathrm{ii}-1}L^{\mathrm{ie}})_{k,k}
		=  \tilde{\Psi}\sum_{k \in f} \alpha_f
		- \sum_{\ell'\neq k} \tilde{\Psi}_G^{k\ell',(m)_{m\neq k,\ell'}}.
	\end{equation*}
\end{itemize}
According to \eqref{Ldef}, the contributions proportional to $\tilde{\Psi}$ cancel in both cases when we subtract $(\tilde{\Psi}L^{\mathrm{ee}})_{k,\ell}$ from \eqref{term1}, such that \eqref{temp} becomes
\begin{equation*}
	\sI
	= \tilde{\Psi}^{-d/2}\exp\Big(
		-\tilde{\Psi}^{-1}
		\sum_{\mathclap{1\leq k<\ell\leq V^\Ext}}
			(x_k^2-2x_k^t x_\ell+x_\ell^2)
			\tilde{\Psi}_G^{k\ell,(m)_{m\neq k,\ell}}
	\Big)
	=
	\tilde{\Psi}^{-d/2}\exp(-\tilde{\Phi}_G/\tilde{\Psi}).
\end{equation*}
Let us now insert a factor $1=\int_0^{\infty} \delta(t-H^{1/r}(\alpha))\dd t$ into \eqref{f1}, where $H(\alpha)$ can be any homogeneous polynomial of degree $r>0$ which is positive inside $\Delta$.
After we substitute all $\alpha_e$ by $t\alpha_e$ and collect the powers of $t$, the integrand of \eqref{f1} becomes
\begin{equation*}
	\delta(1-H^{1/r}(\alpha))
	\Big(\prod_{e}\alpha_e^{n_e+\lambda\nu_e-1}\partial_{\alpha_e}^{n_e}\Big)
	\tilde{\Psi}^{-d/2}
	\Big[
		\int_0^{\infty}
	t^{M_G-1}
	\mathrm{e}^{-t\tilde{\Phi}_G/\tilde{\Psi}}
	\dd t
	\Big]
	\prod_e\dd\alpha_e,
\end{equation*}
because $\tilde{\Psi}$ and $\tilde{\Phi}_G$ are homogeneous in $\alpha$ of degree $V^{\Int}$ and $V^{\Int}+1$, respectively.
We integrate over $t$ using \eqref{trick}. The choice $H(\alpha)=\alpha_e$ for some edge $e$ gives a particularly simple representation as an affine integral over $\RR_+^{E_G-1}$ which is equivalent to \eqref{fdualparam}.

\section{proof of Theorem~\ref{generalthm}}
In this section we prove the real analyticity of graphical functions. Because the polynomial $\tilde{\Phi}_G$ from \eqref{eq:dual-phi} depends on the squared distances
\begin{equation*}
	s_{i,j} = \norm{x_i-x_j}^2
\end{equation*}
between the external vertices, we may use the dual parametric representation \eqref{fdualparam} to define $f_G^{(\lambda)}(s)$ as a function of the vector $s=(s_{i,j})_{1\leq i<j\leq V^\Ext}$. In the (simply connected) domain where all components of $s$ have positive real parts, the integral \eqref{fdualparam} remains absolutely convergent and hence $f_G^{(\lambda)}(s)$ an analytic function of $s$:
\begin{thm}\label{analytic}
	Let $G$ be a graph with a convergent graphical function \eqref{eq:gf-xspace}.
Then $f_G^{(\lambda)}(x)$ extends to a single-valued, analytic function 
\begin{equation*}
	f_G^{(\lambda)}(s) \colon
	\big\{
		s\in \CC^{V^\Ext(V^\Ext-1)/2}\colon
		\text{$\RE s_{i,j}>0$ for all $1\leq i<j\leq V^{\Ext}$}
	\big\}
	\longrightarrow \CC.
\end{equation*}
\end{thm}
In the special case of three external vertices, this implies the real analyticity of $f_G^{(\lambda)}(z)$ on $\CC\setminus\set{0,1}$:
\begin{proof}[Proof of Theorem~\ref{generalthm}]
	Let $z\in\CC\setminus\set{0,1}$. For the three external labels $0$, $1$, $z$ we have $s_{0,1}=1>0$, $s_{0,z}=z\zz>0$ and $s_{1,z}=(z-1)(\zz-1)>0$ according to \eqref{externalvectors}. With theorem~\ref{analytic} we see that $f_G^{(\lambda)}(z,\bar{z})$ is composition of analytic functions, which proves (G3).
	The identity (G1) is immediate from \eqref{fdualparam} as it expresses $f_G^{(\lambda)}(z)$ as a function of $\abs{z}$ and $\abs{1-z}$.
	Finally recall that $f_G^{(\lambda)}(z)$ is defined as the value of the (convergent) integral \eqref{eq:gf-xspace} and thus manifestly single-valued.
\end{proof}
For the proof of Theorem~\ref{analytic} we need the following notation:
\begin{defn}
	Let $g$ be a subgraph of $G$ with edge set $\sE_g\subseteq\sE_G$ and let $Q\in\CC[\alpha_e,e\in\sE_G]$ be a polynomial in the edge variables of $G$.
Then, the (low) degree ($\ldeg{g}(Q)$) $\deg_g(Q)$ of $Q$ is the (low) degree of $Q$ in the edge variables $\alpha_e$, $e\in\sE_g$ of the subgraph $g$.

	In other words, $c=\ldeg{g}(Q)$ is the largest integer such that each monomial in $Q$ has at least $c$ factors $\alpha_e$ with $e \in \sE_g$ (with multiplicity).
Similarly, $C=\deg_g(Q)$ is the smallest integer such that each monomial in $Q$ has at most $C$ factors $\alpha_e$ with $e\in \sE_g$.
\end{defn}
Note that $\ldeg{g}(Q)$ and $\deg_g(Q)$ are defined for polynomials $Q$ in $E_G$ variables. So for $Q=\alpha_1-\alpha_3+\alpha_2$ we have $\ldeg{\set{2}}(Q)=0$, even though on
the subspace $\alpha_1=\alpha_3$ the low degree of $Q$ in $\alpha_2$ is 1.
\begin{prop}\label{prop1}
Let $g$ be a subgraph of a graph $G$ with external vertices. Let $\tilde{\Psi}^p_G(\alpha)$ be a dual spanning forest polynomial \eqref{eq:dual-forpol}
for some partition $p$ of external vertices. Then
\begin{equation}\label{subgraphineq}
	\ldeg{g}(\tilde{\Psi}^p_G)\geq V^{\Int}_g,\quad\deg_g(\tilde{\Psi}^p_G)
	\leq V_g-1,
\end{equation}
where $\sV_g$ and $V^{\Int}_g$ are as in \eqref{ultraviolet} and \eqref{infrared}, respectively.
\end{prop}
\begin{proof}
Let $F\in\sF^p_G$ be a spanning forest of $G$. In every tree $T$ of $F$ we choose an external vertex $v_T\in T$ and we orient all edges of $T$ such that they point towards $v_T$.
Because $F$ is spanning, every $g$-internal vertex $u$ has one outgoing edge in $F$. Conversely, every edge in $F$ has a unique vertex $u$ as source, therefore
\begin{equation*}
	\ldeg{g}(\tilde{\Psi}^p_G)
	=\min_{F\in\sF^p_G} E_{g\cap F}
	\geq V^\Int_g.
\end{equation*}
Finally we use that $g\cap F$ is a forest in $g$ and thus has at most $V_g-1$ edges, hence
\begin{equation*}
	\deg_g(\tilde{\Psi}^p_G)
	=\max_{F\in\sF^p_G} E_{g\cap F}
	=V_g-1
	.\qedhere
\end{equation*}
\end{proof}
\begin{proof}[Proof of Theorem~\ref{analytic}]
We first derive Theorem~\ref{analytic} from \eqref{fdualparam} in the case that all $n_e=0$.
We consider the integrand as a function of the vector $s=(s_{i,j})_{i,j \in \sV^\Ext,i< j}$ which we restrict to the complex domain ($\varepsilon>0$ may be chosen arbitrarily small)
	\begin{equation*}
		\Omega^{\varepsilon}
		= \left\{
			s \colon 
			\RE s_{i,j} > \varepsilon
			\quad\text{for all $1\leq i<j\leq V^{\Ext}$}
		\right\}
		\subset \CC^{V^{\Ext}(V^{\Ext}-1)/2}.
	\end{equation*}
	Let $\hat{s}_{i,j} = \norm{\hat{x}_i-\hat{x}_j}^2$ denote the distances of an arbitrary set $\hat{x}\in \RR^{d V^\Ext}$ of pairwise distinct points. We can rescale $\hat{x}$ to ensure $\max_{i<j} (\hat{s}_{i,j}) = \varepsilon$, such that
\begin{equation*}
	\abs{\tilde{\Phi}_G(\alpha,s)}
	\geq\RE \tilde{\Phi}_G(\alpha,s)
	> \tilde{\Phi}_G(\alpha,\hat{x})
\end{equation*}
for every $s \in \Omega^{\varepsilon}$ and all $\alpha \in \RR_+^E$. As $f^{(\lambda)}_G(\hat{x})$ is convergent, its parametric integrand provides an integrable majorant
$F(\alpha,\hat{s}) \geq F(\alpha,s)$
to the integrand $F(\alpha,s)$ of $f_G^{(\lambda)}(s)$, uniformly for all $s \in \Omega^{\varepsilon}$.
This implies the analyticity of $f_G^{(\lambda)}(s)$ in $\Omega^{\varepsilon}$, for every $\varepsilon>0$ (we cite this result below as theorem~\ref{holomorphic}).

Now let us remove the restriction that $n_e=0$. We compute the derivatives in \eqref{fdualparam} and write the resulting integrand as
\begin{equation}\label{eq:integrand-differentiated}
	F(\alpha,s) =
	\left[ \prod_e \alpha_e^{n_e+\lambda \nu_e-1} \right]
	\frac{\sum_{m} \alpha^m q_m(s)}{\tilde{\Phi}_G(\alpha,s)^{M_G+\sum_e n_e}\tilde{\Psi}(\alpha)^{d/2-M_G+\sum_e n_e}},
\end{equation}
where we expanded the numerator polynomial into its monomials $\alpha^m = \prod_e \alpha_e^{m_e}$ in Schwinger parameters and their coefficients $q_m \in \QQ[s_{i,j}]$.
Note that the operators $\alpha_e \partial_{\alpha_e}$ do not change the $\alpha$-degree, so $F$ stays homogeneous of degree $-E_G$ in the $\alpha$ variables, no matter which values are chosen for the $n_e$. This gives
\begin{equation*}
	\sum_e m_e
	=
	\left(
		\deg_G(\tilde{\Phi}_G)+\deg_G(\tilde{\Psi}_G) - 1
	\right)\sum_e n_e
	=2V^{\Int}\sum_e n_e
	,
\end{equation*}
because the polynomials $\tilde{\Phi}_G$ and $\tilde{\Psi}_G$ have the $\alpha$-degrees $V^{\Int}+1$ and $V^{\Int}$.
If we write \eqref{eq:integrand-differentiated} as $F(\alpha,s) = \sum_m q_m(s) F_m(\alpha,s)$ we can thus identify each $F_m$ with the (dual parametric) integrand for $f^{(\lambda')}_G(s)$ in 
$d'=2\lambda'+2=d+4\sum_e n_e$
dimensions with weights
$\lambda' \nu_e' = \lambda \nu_e + n_e + m_e > 0$. With the first part of the proof it suffices to show that each of these $f^{(\lambda')}_G$ is a convergent graphical function.
We therefore have to consider the infrared \eqref{infrared} and ultraviolet \eqref{ultraviolet} conditions.
Because differentiation $\partial_{\alpha_e}$ for $e\in\sE_g$ can lower the low degree by at most one, we obtain
\begin{equation*}
\sum_{e\in g}m_e-(\ldeg{g}(\tilde{\Phi}_G)+\ldeg{g}(\tilde{\Psi}_G))\sum_{e\in G}n_e\geq-\sum_{e\in g}n_e.
\end{equation*}
From the convergence of $f^{(\lambda)}_G$ and from Proposition~\ref{prop1} we obtain
\begin{equation*}
	\sum_{e\in g}\lambda'\nu_e'
	=\sum_{e\in g}(\lambda\nu_e+n_e+m_e)
	>\tfrac{d}{2}V_g^{\Int}+2V_g^{\Int}\sum_{e\in G}n_e
	=\tfrac{d'}{2} V_g^{\Int},
\end{equation*}
proving infrared convergence. Likewise, differentiation $\partial_{\alpha_e}$ for $e\in\sE_g$ lowers the degree by at least one, yielding
\begin{equation*}
	\sum_{e\in g}m_e-(\deg_g(\tilde{\Phi}_G)+\deg_g(\tilde{\Psi}_G))\sum_{e\in G}n_e\leq-\sum_{e\in g}n_e.
\end{equation*}
Together with Proposition~\ref{prop1} this proves ultraviolet convergence (and thus completes our proof of Theorem~\ref{analytic}):
\begin{equation*}
	\sum_{e\in g}\lambda'\nu_e'
	=\sum_{e\in g}(\lambda\nu_e+n_e+m_e)
	<\left(\tfrac{d}{2}+2\sum_{e\in G}n_e \right)(V_g-1)
	=\tfrac{d'}{2}(V_g-1)
	.\qedhere
\end{equation*}
\end{proof}
For convenience of the reader we cite here the result from calculus in the form \cite[Theorem~2.12]{Sauvigny}, which is perfectly adapted to our application:
\begin{thm}\label{holomorphic}
Let $\Theta\subset\RR^m$ and $\Omega\subset\CC^n$ denote domains in the respective spaces of dimensions $m,n\in\NN$. Furthermore, let
\begin{equation*}
	f(t,z)
	=f(t_1,\ldots,t_m,z_1,\ldots,z_n)\colon
	\Theta\times\Omega\longrightarrow\CC
\end{equation*}
represent a continuous function with the following properties:
\begin{itemize}
	\item For each fixed $t\in\Theta$, the function
$ z \mapsto f(t,z) $
is holomorphic in $z\in\Omega$.

	\item We have a continuous function 
$
	F(t)\colon\Theta\longrightarrow[0,\infty)
$ which is integrable,
\begin{equation*}
	\int_\Theta F(t)\ \dd t< \infty,
\end{equation*}
and uniformly majorizes $f$:
$ \abs{f(t,z)} \leq F(t)$ for all $(t,z)\in\Theta\times\Omega$.
\end{itemize}
Then the function
$
	z \mapsto \int_\Theta f(t,z)\ \dd t
$
is holomorphic in $\Omega$.
\end{thm}
\begin{remark}
We may consider a graphical function $f_G^{(\lambda)}(z)$ as a function of two complex variables $z$ and $\zz$ and analytically continue away from the locus where $\zz$
is the complex conjugate of $z$. In this case Theorem~\ref{analytic} states that $f_G^{(\lambda)}$ is analytic in $z$ and $\zz$ if $\RE z\zz>0$ and $\RE (1-z)(1-\zz)>0$.
If one continues analytically beyond this domain, additional singularities will in general appear.
Already in example~\ref{ex:g4} we encounter $z=\bar{z}$, which corresponds to the vanishing of the K\"{a}ll\'{e}n function
	\begin{equation*}
		(z-\zz)^2
		=
		s_{0,z}^2+s_{1,z}^2+s_{0,1}^2-2 s_{0,z} s_{1,z} - 2 s_{0,z} s_{0,1} - 2 s_{1,z} s_{0,1}.
	\end{equation*}
For bigger graphs the singularity structure outside $\RE z\zz>0$, $\RE (1-z)(1-\zz) >0 $ becomes more and more complicated (see \cite[table~1]{Panzer:ManyScales} for a few examples).
\end{remark}

\section{proof of Theorem~\ref{planarthm}}
Planar duality for graphical functions is specific to three external labels for which we use $0$, $1$, $z$. Let us first recall the notion of planarity and planar duality for Feynman graphs in this case.\footnote{%
In the physics literature, this definition is standard \cite{MinamiNakanishi}. We did not find an established name for this in the literature on graph theory, except for the term ``circular planar graph'' used in \cite{CurtisIngermanMorrow}.}

\begin{defn}\label{dualdef}
Let $G$ be a graph with three external labels $0$, $1$, $z$. Let $G_v$ be the graph that we obtain from $G$ by adding an extra vertex $v$ which is connected to the external vertices
of $G$ by edges $\{0,v\}$, $\{1,v\}$, $\{z,v\}$, respectively. We say that $G$ is \emph{externally planar} if and only if $G_v$ is planar.

Let $G_v$ be planar and $G_v^\star$ a planar dual of $G_v$. The edges $e^\star$ of $G_v^\star$ are in one to one correspondence with the edges $e$ of $G_v$.
A planar dual of $G$ is given by $G_v^\star$ minus the triangle $\{0,v\}^\star$, $\{1,v\}^\star$, $\{z,v\}^\star$ with external labels $0$, $1$, $z$ corresponding to the
faces $1zv$, $0zv$, $01v$, respectively. The edge weights of $G^\star$ are related to the edge weights of $G$ by \eqref{weightsstar}: $\lambda\nu_e+\lambda\nu_{e^\star}=d/2$.
\end{defn}
We can draw an externally planar graph $G$ with the external labels 0, 1, $z$ in the outer face. A dual $G^\star$ then has also the labels in the outer face, `opposite' to the labels of $G$, see Figure~\ref{fig:duals}.

Another way to construct this dual is by adding three edges $e_{01}=\{0,1\}$, $e_{0z}=\set{0,z}$, $e_{1z}=\set{1,z}$ to $G$ to obtain a graph $G_e$. Its dual $G_e^\star$ differs from $G_v^\star$ upon replacing the triangle $\set{0,v}^\star$, $\set{1,v}^\star$, $\set{z,v}^\star$ by a star $e_{01}^\star$, $e_{0z}^\star$, $e_{1z}^\star$.
From $G_e^\star$ we obtain $G^\star$ by removing this star and labeling its tips with $z$, $1$, $0$, respectively.
Clearly both constructions (starting from the same planar embedding of $G$) lead to the same dual $G^\star$ and prove
\begin{lem}
	Let $G$ be externally planar with dual $G^\star$. Then $G^\star$ is externally planar and $G$ is a dual of $G^\star$.
\end{lem}

\begin{proof}[Proof of Theorem~\ref{planarthm}]
	Because the edge weights are positive we can use $n_e=0$ in \eqref{fdualparam}. From $M_G=d/2$ we obtain (see \eqref{defM} and \eqref{weightsstar})
\begin{equation*}
	M_{G^\star}
	=\sum_e(\tfrac{d}{2}-\lambda\nu_e)-\tfrac{d}{2}V^{\Int}_{G^\star}
	= \tfrac{d}{2}(E_G-V^{\Int}_{G^\star}-V^{\Int}_G-1)
	= \tfrac{d}{2}(E_{G_v}-V_{G_v^\star}-V_{G_v}+3)
\end{equation*}
where $E_G$ is the number of edges of $G$.
%$E_G = E_{G_v}-3$, $V_{G^{\star}}^\Int = V_{G_v^\star} - 3$ and $V^{\Int}_G=V_{G_v}-4$ 
As the vertices of $G_v^\star$ are the faces of the planar embedding of $G_v$, Euler's formula for planar graphs shows $M_{G^\star}=d/2$.

Comparing \eqref{fdualparam} for the graph $G$ with \eqref{fparam} for the graph $G^\star$ leads to \eqref{eq:planar-dual} if we identify $\alpha_e=\alpha_{e^\star}$ for all edges $e$, provided that 
$
	\tilde{\Phi}_G=\Phi_{G^\star}
$.
This amounts to the identity
$\tilde{\Psi}^{ij,k}_G=\Psi^{ij,k}_{G^\star}$
of spanning forest polynomials for all triples $\{i,j,k\}=\{0,1,z\}$ and hence follows from the bijection of 2-forests given by
\begin{equation*}
	\sF^{ij,k}_G \ni F
	\longleftrightarrow 
	F^\star := \{e^\star\colon e\not\in F\}\in\sF^{ij,k}_{G^\star}.
\end{equation*}
Namely, for any given $F\in\sF^{ij,k}_G$ consider the spanning tree $T_i = F \cup \set{\set{i,v},\set{k,v}}$ of $G_v$. As Tutte points out \cite[Theorem~2.64]{Tutte}, its dual
$T_i^\star = \set{e^{\star}\colon e\notin T}\subseteq \sE_{G_v^\star}$
is a spanning tree of $G_v^\star$ and therefore,
$F^{\star} = T_i^\star\setminus\set{j,v}^\star$
is indeed a $2$-forest.
Furthermore, the edge $\set{j,v}^\star$ connects the external vertices $i$ and $k$ of $G^\star$ and thus $F^\star$ cannot connect $i$ with $k$ (otherwise, $T_i^\star = F^\star \cup \set{j,v}^\star$ would contain a cycle). Likewise (interchanging $i$ and $j$) $F^\star$ does not connect $j$ with $k$, hence
$F^{\star} \in \sF^{i,k}_{G^\star} \cap \sF^{j,k}_{G^\star} = \sF^{ij,k}_{G^\star}$. 
Finally, the symmetry $F=(F^\star)^\star$ implies that the map $F\mapsto F^\star$ is injective and onto.
\end{proof}

\bibliographystyle{JHEPsortdoi}
\bibliography{refs}

\providecommand{\href}[2]{#2}\providecommand{\eprintlink}[2]{\href{#1}{#2}}\begingroup\begin{thebibliography}{10}

\bibitem{BEK}
S.~Bloch, H.~Esnault and D.~Kreimer,
  \href{\detokenize{http://dx.doi.org/10.1007/s00220-006-0040-2}}{\textit{On
  motives associated to graph polynomials}},
  \href{\detokenize{http://dx.doi.org/10.1007/s00220-006-0040-2}}{\emph{Commun.
  Math. Phys.} \textbf{267} (2006), no.~1 }pp.~181--225,
  \eprintlink{http://arxiv.org/abs/math/0510011}{arXiv:math/0510011}.

\bibitem{BognerWeinzierl}
C.~Bogner and S.~Weinzierl,
  \href{\detokenize{http://dx.doi.org/10.1142/S0217751X10049438}}{\textit{Feynman
  graph polynomials}},
  \href{\detokenize{http://dx.doi.org/10.1142/S0217751X10049438}}{\emph{International
  Journal of Modern Physics A} \textbf{25} (2010) }pp.~2585--2618,
  \eprintlink{http://arxiv.org/abs/1002.3458}{arXiv:1002.3458 [hep-ph]}.

\bibitem{Brown:TwoPoint}
F.~C.~S. Brown,
  \href{\detokenize{http://dx.doi.org/10.1007/s00220-009-0740-5}}{\textit{The
  massless higher-loop two-point function}},
  \href{\detokenize{http://dx.doi.org/10.1007/s00220-009-0740-5}}{\emph{Commun.
  Math. Phys.} \textbf{287} (May, 2009) }pp.~925--958,
  \eprintlink{http://arxiv.org/abs/0804.1660}{arXiv:0804.1660 [math.AG]}.

\bibitem{Brown:SomeFI}
F.~C.~S. Brown, \href{\detokenize{http://arxiv.org/abs/0910.0114}}{``On the
  periods of some {Feynman} integrals.''} preprint, Oct., 2009,
  \eprintlink{http://arxiv.org/abs/0910.0114}{arXiv:0910.0114 [math.AG]}.

\bibitem{ZigZag}
F.~C.~S. Brown and O.~Schnetz,
  \href{\detokenize{http://dx.doi.org/10.1016/j.jnt.2014.09.007}}{\textit{Single-valued
  multiple polylogarithms and a proof of the zig-zag conjecture}},
  \href{\detokenize{http://dx.doi.org/10.1016/j.jnt.2014.09.007}}{\emph{Journal
  of Number Theory} \textbf{148} (Mar., 2015) }pp.~478--506,
  \eprintlink{http://arxiv.org/abs/1208.1890}{arXiv:1208.1890 [math.NT]}.

\bibitem{BrownYeats:WD}
F.~C.~S. Brown and K.~A. Yeats,
  \href{\detokenize{http://dx.doi.org/10.1007/s00220-010-1145-1}}{\textit{Spanning
  forest polynomials and the transcendental weight of {Feynman} graphs}},
  \href{\detokenize{http://dx.doi.org/10.1007/s00220-010-1145-1}}{\emph{Commun.
  Math. Phys.} \textbf{301} (Jan., 2011) }pp.~357--382,
  \eprintlink{http://arxiv.org/abs/0910.5429}{arXiv:0910.5429 [math-ph]}.

\bibitem{Chaiken}
S.~Chaiken, \href{\detokenize{http://dx.doi.org/10.1137/0603033}}{\textit{A
  combinatorial proof of the all minors matrix tree theorem}},
  \href{\detokenize{http://dx.doi.org/10.1137/0603033}}{\emph{SIAM Journal on
  Algebraic Discrete Methods} \textbf{3} (Sept., 1982) }pp.~319--329.

\bibitem{CurtisIngermanMorrow}
E.~B. Curtis, D.~Ingerman and J.~A. Morrow,
  \href{\detokenize{http://dx.doi.org/10.1016/S0024-3795(98)10087-3}}{\textit{Circular
  planar graphs and resistor networks}},
  \href{\detokenize{http://dx.doi.org/10.1016/S0024-3795(98)10087-3}}{\emph{Linear
  Algebra and its Applications} \textbf{283} (Nov., 1998) }pp.~115--150.

\bibitem{OffShellConformal}
J.~Drummond, C.~Duhr, B.~Eden, P.~Heslop, J.~Pennington and V.~A. Smirnov,
  \href{\detokenize{http://dx.doi.org/10.1007/JHEP08(2013)133}}{\textit{Leading
  singularities and off-shell conformal integrals}},
  \href{\detokenize{http://dx.doi.org/10.1007/JHEP08(2013)133}}{\emph{JHEP}
  \textbf{8} (Aug., 2013) }p.~133,
  \eprintlink{http://arxiv.org/abs/1303.6909}{arXiv:1303.6909 [hep-th]}.

\bibitem{Drummond:Ladders}
J.~M. Drummond,
  \href{\detokenize{http://dx.doi.org/10.1007/JHEP02(2013)092}}{\textit{Generalised
  ladders and single-valued polylogs}},
  \href{\detokenize{http://dx.doi.org/10.1007/JHEP02(2013)092}}{\emph{Journal
  of High Energy Physics} \textbf{2013} (Feb., 2013) }p.~92,
  \eprintlink{http://arxiv.org/abs/1207.3824}{arXiv:1207.3824 [hep-th]}.

\bibitem{Elstrodt}
J.~Elstrodt,
  \href{\detokenize{http://dx.doi.org/10.1007/978-3-642-17905-1}}{\emph{{Ma\ss-
  und Integrationstheorie}}}.
\newblock Springer-Lehrbuch. Berlin: Springer, 7~ed., 2011.

\bibitem{ItzyksonZuber}
C.~Itzykson and J.-B. Zuber, {\emph{Quantum Field Theory}}.
\newblock Dover Publications, Inc., 2006.
\newblock first published by McGraw-Hill in 1980.

\bibitem{LamLebrun}
C.~S. Lam and J.~P. Lebrun,
  \href{\detokenize{http://dx.doi.org/10.1007/BF02753153}}{\textit{Feynman-parameter
  representations for momentum- and configuration-space diagrams}},
  \href{\detokenize{http://dx.doi.org/10.1007/BF02753153}}{\emph{Il Nuovo
  Cimento A Series 10} \textbf{59} (Feb., 1969) }pp.~397--421.

\bibitem{LowensteinZimmermann}
J.~H. Lowenstein and W.~Zimmermann,
  \href{\detokenize{http://dx.doi.org/10.1007/BF01609059}}{\textit{The power
  counting theorem for {Feynman} integrals with massless propagators}},
  \href{\detokenize{http://dx.doi.org/10.1007/BF01609059}}{\emph{Commun. Math.
  Phys.} \textbf{44} (1975), no.~1 }pp.~73--86.

\bibitem{Mattner}
L.~Mattner,
  \href{\detokenize{http://www.nieuwarchief.nl/serie5/pdf/naw5-2001-02-1-032.pdf}}{\textit{Complex
  differentiation under the integral}},
  \href{\detokenize{http://www.nieuwarchief.nl/serie5/pdf/naw5-2001-02-1-032.pdf}}{\emph{Nieuw.
  Arch. Wisk.} \textbf{5/2} (Mar., 2001) }pp.~32--35.

\bibitem{MinamiNakanishi}
M.~Minami and N.~Nakanishi,
  \href{\detokenize{http://dx.doi.org/10.1143/PTP.40.167}}{\textit{Duality
  between the {Feynman} integral and a perturbation term of the {Wightman}
  function}},
  \href{\detokenize{http://dx.doi.org/10.1143/PTP.40.167}}{\emph{Progress of
  Theoretical Physics} \textbf{40} (July, 1968) }pp.~167--177.

\bibitem{Nakanishi:Hankel}
N.~Nakanishi,
  \href{\detokenize{http://dx.doi.org/10.1143/PTP.42.966}}{\textit{Feynman-parametric
  formula for the {Hankel}-transformed position-space {Feynman} integral}},
  \href{\detokenize{http://dx.doi.org/10.1143/PTP.42.966}}{\emph{Progress of
  Theoretical Physics} \textbf{42} (Oct., 1969) }pp.~966--977.

\bibitem{Nakanishi:Book}
N.~Nakanishi, {\emph{{Graph} theory and {Feynman} integrals}}, vol.~11 of
  \emph{Mathematics and its applications}.
\newblock Gordon and Breach, New York, 1971.

\bibitem{Panzer:PhD}
E.~Panzer, \href{\detokenize{http://arxiv.org/abs/1506.07243}}{\emph{Feynman
  integrals and hyperlogarithms}}.
\newblock PhD thesis, Humboldt-Universit{\"a}t zu Berlin, 2014.
\newblock \eprintlink{http://arxiv.org/abs/1506.07243}{arXiv:1506.07243
  [math-ph]}.

\bibitem{Panzer:ManyScales}
E.~Panzer,
  \href{\detokenize{http://dx.doi.org/10.1007/JHEP03(2014)071}}{\textit{On
  hyperlogarithms and {Feynman} integrals with divergences and many scales}},
  \href{\detokenize{http://dx.doi.org/10.1007/JHEP03(2014)071}}{\emph{JHEP}
  \textbf{2014} (Mar., 2014) }p.~71,
  \eprintlink{http://arxiv.org/abs/1401.4361}{arXiv:1401.4361 [hep-th]}.

\bibitem{Panzer:HyperInt}
E.~Panzer,
  \href{\detokenize{http://dx.doi.org/10.1016/j.cpc.2014.10.019}}{\textit{Algorithms
  for the symbolic integration of hyperlogarithms with applications to
  {Feynman} integrals}},
  \href{\detokenize{http://dx.doi.org/10.1016/j.cpc.2014.10.019}}{\emph{Computer
  Physics Communications} \textbf{188} (Mar., 2015) }pp.~148--166,
  \eprintlink{http://arxiv.org/abs/1403.3385}{arXiv:1403.3385 [hep-th]}.

\bibitem{Coaction}
E.~Panzer and O.~Schnetz,
  \href{\detokenize{http://arxiv.org/abs/1603.04289}}{``The {Galois} coaction
  on {$\phi^4$} periods.''} to appear in Communications in Number Theory and
  Physics, 2016, \eprintlink{http://arxiv.org/abs/1603.04289}{arXiv:1603.04289
  [hep-th]}.

\bibitem{Sauvigny}
F.~Sauvigny,
  \href{\detokenize{http://dx.doi.org/10.1007/978-1-4471-2981-3}}{\emph{Partial
  differential equations 1}}.
\newblock Springer, Berlin, second~ed., 2012.

\bibitem{GraphicalFunctions}
O.~Schnetz,
  \href{\detokenize{http://dx.doi.org/10.4310/CNTP.2014.v8.n4.a1}}{\textit{Graphical
  functions and single-valued multiple polylogarithms}},
  \href{\detokenize{http://dx.doi.org/10.4310/CNTP.2014.v8.n4.a1}}{\emph{Communications
  in Number Theory and Physics} \textbf{8} (2014), no.~4 }pp.~589--675,
  \eprintlink{http://arxiv.org/abs/1302.6445}{arXiv:1302.6445 [math.NT]}.

\bibitem{Schnetz6loops}
O.~Schnetz, \href{\detokenize{http://arxiv.org/abs/1606.08598}}{``Numbers and
  functions in quantum field theory.''} preprint, 2016,
  \eprintlink{http://arxiv.org/abs/1606.08598}{arXiv:1606.08598 [hep-th]}.

\bibitem{Todorov}
I.~Todorov,
  \href{\detokenize{http://dx.doi.org/10.1007/978-4-431-55285-7_10}}{\textit{Polylogarithms
  and multizeta values in massless {Feynman} amplitudes}},  in \emph{Lie Theory
  and Its Applications in Physics} (V.~Dobrev, ed.), vol.~111 of \emph{Springer
  Proceedings in Mathematics \& Statistics}, pp.~155--176.
\newblock Springer Japan, 2014.

\bibitem{Tutte}
W.~T. Tutte,
  \href{\detokenize{http://dx.doi.org/10.6028/jres.069B.001}}{\textit{Lectures
  on matroids}},
  \href{\detokenize{http://dx.doi.org/10.6028/jres.069B.001}}{\emph{J. Res.
  Nat. Bur. Standards Sect. B} \textbf{69B} (1965), no.~1--2 }pp.~1--47.
  reprinted in: Selected papers of W.~T. Tutte, Vol. II (D.\ McCarthy, R.~G.\
  Stanton, eds.), Charles Babbage Research Centre, St. Pierre, Manitoba, 1979,
  pp. 439--496.

\end{thebibliography}\endgroup
\end{document}